\documentclass{llncs}
\usepackage{amsmath}
\usepackage{amssymb}
\usepackage{mathpartir}
\usepackage{graphicx}
\usepackage{xspace}
\usepackage{listings}

\newcommand{\insV}{\mathsf{insV}}
\newcommand{\insE}{\mathsf{insE}}
\newcommand{\insP}{\mathsf{insP}}
\newcommand{\delV}{\mathsf{delV}}
\newcommand{\delE}{\mathsf{delE}}
\newcommand{\delP}{\mathsf{delP}}
\newcommand{\updP}{\mathsf{updP}}
\newcommand{\grantnum}[2]{#2\xspace}
\newcommand{\refapp}[1]{Appendix~\ref{#1}}
\begin{document}

\title{Flexible graph matching and graph edit distance using answer set programming}
\author{Sheung Chi Chan\inst{1}
\and James Cheney\inst{1,2}}

\institute{University of Edinburgh, UK
\and
The Alan Turing Institute}
\maketitle

\begin{abstract}
  The \emph{graph isomorphism}, \emph{subgraph isomorphism}, and
  \emph{graph edit distance} problems are combinatorial problems with
  many applications. Heuristic exact and approximate algorithms for
  each of these problems have been developed for different kinds of
  graphs: directed, undirected, labeled, etc.  However, additional
  work is often needed to adapt such algorithms to different classes
  of graphs, for example to accommodate both labels and property
  annotations on nodes and edges.  In this paper, we propose an
  approach based on answer set programming.  We show how each of these
  problems can be defined for a general class of \emph{property
    graphs} with directed edges, and labels and key-value properties
  annotating both nodes and edges.  We evaluate this approach on a
  variety of synthetic and realistic graphs, demonstrating that it is
  feasible as a rapid prototyping approach.
\end{abstract}

\section{Introduction}

Graphs are a pervasive and widely applicable data structure in
computer science.  To name just a few examples, graphs can represent
symbolic knowledge structures extracted from Wikipedia~\cite{dbpedia}, provenance
records describing how a computer system executed to produce a result~\cite{pasquier17socc},
or chemical structures in a scientific knowledge base~\cite{muta}.  In many
settings, it is of interest to solve \emph{graph matching} problems,
for example to determine when two graphs have the same
structure, or when one graph appears in another, or to measure how
similar two graphs are.

Given two graphs, possibly with labels or other data associated with
nodes and edges, the \emph{graph isomorphism} problem (GI) asks
whether the two graphs have the same structure, that is, whether there
is an invertible mapping from one graph to another that preserves and
reflects edges and any other constraints.  The \emph{subgraph isomorphism}
problem (SUB) asks whether one graph is isomorphic to a subgraph of
another.  Finally, the \emph{graph edit distance} problem (GED) asks whether
one graph can be transformed into another via a sequence of edit steps,
such as insertion, deletion, or updates to nodes or edges.

These are well-studied problems.  Each is in the class NP, with SUB
and GED being NP-complete~\cite{garey79}, while the lower bound of the complexity 
of GI is an open problem~\cite{arvind05beatcs}.  Approximate and exact
algorithms for graph edit distance, based on
heuristics or on reduction to other NP-complete problems, have been
proposed~\cite{gao2010paa,lerouge17pr,riesen15,chen19kbs}.  Moreover, for special cases such as
database querying, there are algorithms for subgraph isomorphism that
can provide good performance in practice when matching small
query subgraphs against graph databases~\cite{lee12vldb}.

However, there are circumstances in which none of the available
techniques is directly suitable.  For example, many of the algorithms
considered so far assume graphs of a specific form, for example with
unordered edges, or unlabeled nodes and edges.  In contrast, many
typical applications use graphs with complex structure, such as
property graphs: directed multigraphs in which nodes and edges can
both be labeled and annotated with sets of key-value pairs
(\emph{properties}).  Adapting an existing algorithm to deal with each
new kind of graph is nontrivial.  Furthermore, some applications
involve searching for isomorphisms, subgraph isomorphisms, or edit
scripts subject to additional constraints~\cite{zampelli05cp,chan19mw}.

In this paper we advocate the use of \emph{answer set programming}
(ASP) to specify and solve these problems.  Property graphs can be
represented uniformly as sets of logic programming facts, and each of
the graph matching problems we have mentioned can be specified using
ASP in a uniform way.  Concretely, we employ the Clingo ASP solver, but
our approach relies only on standard ASP features.  

For each of the problems we consider, it is clear in principle that it
should be possible to encode using ASP, because ASP subsumes the
NP-complete SAT problem.  Our contribution is to show how to encode each of these
problems directly in a way that produces immediately useful results,
rather than via encoding as SAT or other problems and decoding the
results.  For GI and SUB, the encoding is rather direct and the ASP
specifications can easily be read as declarative specifications of the
respective problems; however, the standard formulation of the graph
edit distance problem is not as easy to translate to a logic program
because it involves searching for an edit script whose maximum length
depends on the input.  Instead, we consider an indirect (but still
natural) approach which searches for a partial matching between the
two graphs that minimizes the edit distance, and derives an edit
script (if needed) from this matching.  The proof of correctness of
this encoding is our main technical contribution.

We provide experimental evidence of the practicality of our
declarative approach, drawing on experience with a nontrivial
application: generalizing and comparing provenance
graphs~\cite{chan19mw}.  In this previous work, we needed to solve two
problems: (1) given two graphs with the same structure but possibly
different property values (e.g. timestamps), identify the general
structure common to all of the graphs, and (2) given a background
graph and a slightly larger foreground graph, match the background
graph to the foreground graph and ``subtract'' it, leaving the
unmatched part.  We showed in ~\cite{chan19mw} that our ASP approach
to approximate graph isomorphism and subgraph isomorphism can solve
these problems fast enough that they were not the bottleneck in the
overall system.  In this paper, we conduct further experimental
evaluation of our approach to graph isomorphism, subgraph isomorphism,
and graph edit distance on synthetic graphs and real graphs used in a
recent Graph Edit Distance Contest (GEDC)~\cite{abuaisheh17prl} and
our recent work~\cite{chan19mw}.

\section{Background}

\paragraph{Property graphs}

We consider \emph{(directed) multigraphs} $G = (V,E,src,tgt,lab)$ where $V$ and $E$
are disjoint sets of \emph{node identifiers} and \emph{edge
  identifiers}, respectively, $src,tgt : E \to V$ are functions
identifying the source and target of each edge, and $lab : V\cup E
\to \Sigma$ is a function assigning each vertex and edge a label from
some set $\Sigma$.  Note that multigraphs can have multiple edges with
the same source and target.  Familiar definitions of ordinary directed
or undirected graphs can be recovered by imposing further constraints,
if desired.  

A \emph{property graph} is a directed multigraph extended with an
additional partial function $prop : (V \cup E) \times \Gamma \rightharpoonup \Delta$ where
$\Gamma$ is a set of \emph{keys} and $\Delta$ is a set of \emph{data
  values}.  For the purposes of this paper we assume that all
identifiers, labels, keys and values are represented as Prolog atoms. 

We consider a partial function with range $X$ to be a total function
with range $X \uplus \{\bot\}$ where $\bot$ is a special token not
appearing in $X$.  We consider $X \uplus \{\bot\}$ to be partially ordered
by the least partial order satisfying $\bot \sqsubseteq x$ for all $x
\in X$.

\paragraph{Isomorphisms} 

A \emph{homomorphism} from property graph $G_1$ to $G_2$ is a function $h
: G_1 \to G_2$ mapping $V_1$ to $V_2$ and $E_1 $ to $E_2$, such that:
\begin{itemize}
\item for all $v \in V_1$, 
$lab_2(h(v)) = lab_1(v)$ and $prop_2(h(v),k) \sqsubseteq prop_1(v,k)$ 
\item for all $e \in E_1$, $lab_2(h(e)) = lab_1(e)$ and $prop_2(h(e),k)
  \sqsubseteq prop_1(e,k)$
\item for all $e \in E_1$, $src_2(h(e)) = h(src_1(e))$ and
  $tgt_2(h(e)) = h(tgt_1(e))$
\end{itemize}
(Essentially, $h$ is a pair of functions
$(V_1 \to V_2) \times (E_1 \to E_2)$, but we abuse notation slightly
here by writing $h$ for both.)  As usual, an isomorphism is an
invertible homomorphism whose inverse is also a homomorphism, and
$G_1$ and $G_2$ are isomorphic ($G_1 \cong G_2$) if an isomorphism
between them exists.  Note that the labels of nodes and edges must
match exactly, that is, we regard labels as integral to nodes and
edges, while properties must match only if defined in $G_1$.

\paragraph{Subgraph isomorphism}

A subgraph $G'$ of $G$ is a property graph satisfying:
\begin{itemize}
\item $V' \subseteq V$ and $E' \subseteq E$ 
\item $src'(e) = src(e) \in V'$ and $tgt(e) =
  tgt'(e) \in V'$ for all $e \in E'$
\item $lab'(x) = lab(x)$ when $x \in V' \cup E'$
\item $prop'(x,k) \sqsubseteq prop(x,k) $ when $x \in V'\cup E'$
\end{itemize}
In other words, the vertex and edge sets of $G'$ are subsets of those
of $G$ that still form a meaningful graph, the labels are the same
as in $G'$, and the properties defined in $G'$ are the same as in $G$
(but some properties in $G$ may be omitted).

We say that $G_1$ is \emph{subgraph isomorphic} to $G_2$ ($G_1 \lesssim
G_2$) if there is a
subgraph of $G_2$ to which $G_1$ is isomorphic.  Equivalently, $G_1
\lesssim G_2$ holds if there is a \emph{injective} homomorphism $h :
G_1 \to G_2$.  If such a homomorphism exists, then it maps $G_1$ to an isomorphic
subgraph of $G_2$, whereas if $G_1 \cong G_2' \subseteq G_2$ then the
isomorphism between $G_1$ and $G_2'$ extends to an injective
homomorphism from $G_1$ to $G_2$.

\paragraph{Graph edit distance}

We consider \emph{edit operations}:
\begin{itemize}
\item insertion of a node ($\insV(v,l)$), edge $(\insE(e,v,w,l)$), or
  property $(\insP(x,k,v,d)$)
\item deletion of a node ($\delV(v)$), edge $(\delE(e)$), or
  property $(\delP(x,k)$)
\item in-place update ($\updP(x,k,d)$) of a property value on a
  given node or edge $x$
  with a  given key $k$ to value $d$
\end{itemize}
The meanings of each of these operations are defined in
table~\ref{tab:edit-semantics}, where we write
$G = (V,E,src,tgt,lab,prop)$ for the graph before the edit and
$G' = (V',E',src',tgt',lab',prop')$ for the updated graph.  Each row
of the table describes how each part of $G'$ is defined in terms of
$G$.  In addition, the edit operations have the following
preconditions: Before an insertion, the inserted node, edge, or
property must not already exist; before a deletion, a deleted node
must not be a source or target of an edge, and a node/edge must not
have any properties; before an update, the updated property must
already exist on the affected node or edge.  If these preconditions
are not satisfied, the edit operation is not allowed on $G$.

\begin{table}[tb]
\caption{Edit operation semantics}
\[
\begin{array}{|c|c|c|c|c|c|c|}
\hline
  op & V' & E' & src' & tgt' & lbl' & prop'\\
\hline
\insV(n,l) 
& V \uplus \{v\} & E & src & tgt & lbl [v:=l]  & prop\\
\insE(e,v,w,l) 
& V & E \uplus\{e\} & src[e:=v] & tgt[e:=w] & lbl[e:=l] & prop\\
\insP(x,k,d) 
& V & E & src & tgt & lbl & prop[x,k:=d]\\
\delV(v) 
& V - \{v\} & E & src & tgt & lbl [v:=\bot]  & prop\\
\delE(e) 
& V & E-\{e\} & src[e:=\bot] & tgt[e:=\bot] & lbl[e:=\bot] & prop\\
\delP(x,k) 
& V & E & src & tgt & lbl & prop[x,k:=\bot]\\
\updP(x,k,d) & V & E & src & tgt & lbl & prop[x,k:=d]\\
\hline
\end{array}
\]
\label{tab:edit-semantics}
\end{table}
We write $op(G)$
for the result of $op$ acting on $G$.  More generally, if $ops$ is a list of
operations then we write $ops(G)$ for the result of applying the
operations to $G$.
Given graphs $G_1,G_2$ we define the \emph{graph edit distance}
between $G_1$ and $G_2$ as $GED(G_1,G_2)  = \min\{ |ops| \mid {ops(G_1) = G_2}\}$, that is, the
shortest length of an edit script modifying $G_1$ to $G_2$.

Computing the graph edit distance between two graphs (even without
labels or properties) is an NP-complete problem.  Moreover, we
consider a particular setting where the edit operations all have equal
cost, but in general different weights can be assigned to different
edit operations.  We can consider a slight generalization as follows:
Given a weighting function $w$ mapping edit operations to positive
rational numbers, the \emph{weighted graph edit distance} between $G_1$ and $G_2$
is $wGED(G_1,G_2) = \min\{ \sum_{op\in ops}w(op) \mid {ops(G_1) = G_2}\}$.  The
unweighted graph edit distance is a special case so this
problem is also NP-complete.

\paragraph{Answer set programming}
We assume familiarity with general logic programming concepts
(e.g. familiarity with Prolog or Datalog).  To help make the paper
accessible to readers not already familiar with answer set
programming, we illustrate some programming techniques that differ
from standard logic programming via a short example: coloring the
nodes of an undirected graph with the minimum number of colors.  Graph
3-coloring is a standard example of ASP, but we will adopt a slightly
nonstandard approach to illustrate some key techniques we will rely on
later.  We will use the concrete syntax of the Clingo ASP solver,
which is part of the Potassco
framework~\cite{gebser2011aicom,gebser2018ki}.  Examples given here
and elsewhere in the paper can be run verbatim using the Clingo
interactive online demo\footnote{https://potassco.org/clingo/run/}.

\begin{figure}[tb]
  \centering
\begin{tabular}{cp{5cm}}
  $\vcenter{\hbox{\includegraphics[scale=0.3]{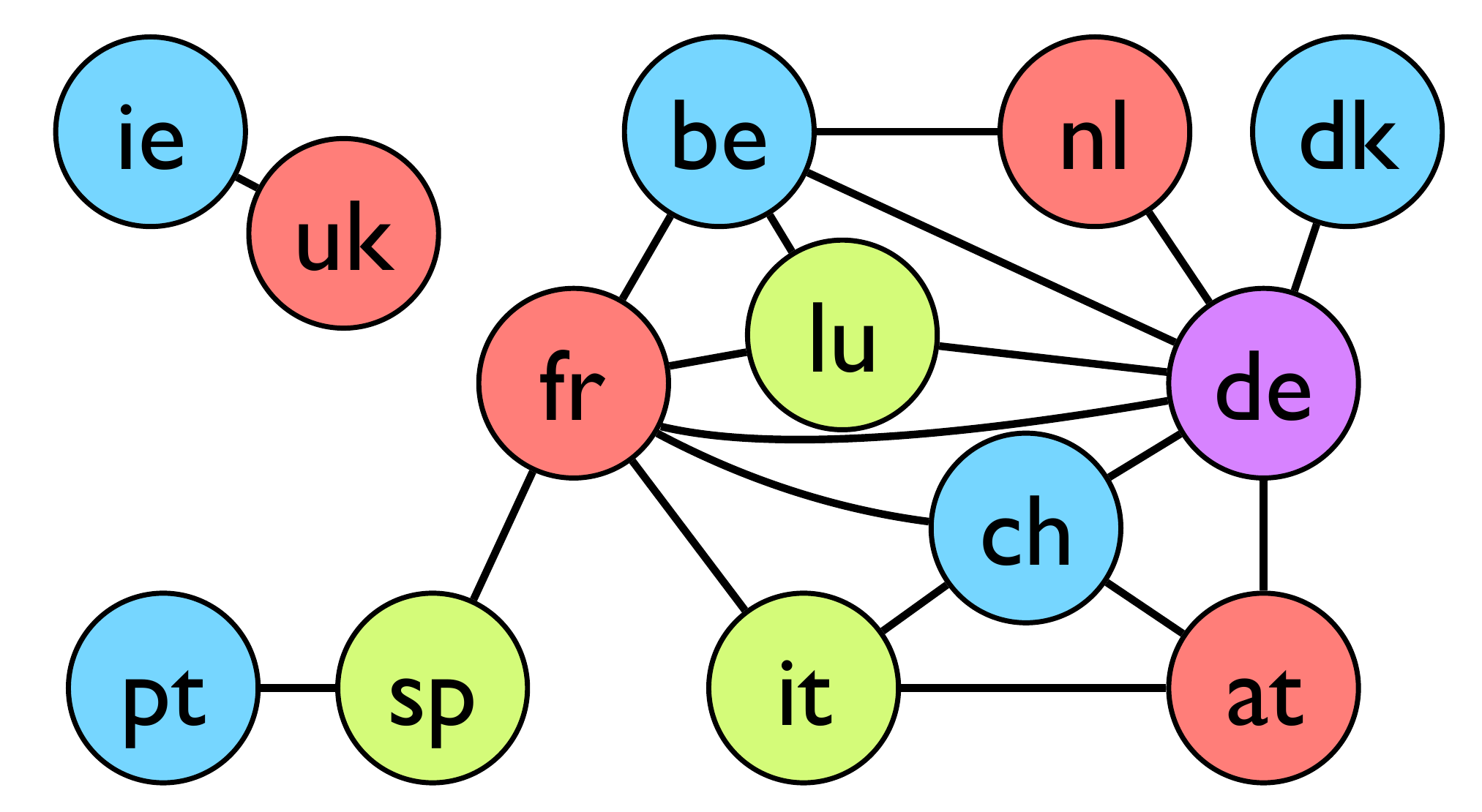}}}$
& \texttt{e(uk,ie). e(fr,(sp;de;ch;it;be;lu)). 
e(sp,pt). 
e(it,at).
e(be,(lu;nl)). e(ch,(it;at)).  e(de,(be;nl;at;dk;ch;lu)).}
\end{tabular}
  \caption{Graph coloring example}
  \label{fig:coloring}
 \lstset{frame=single,numbers=left,caption={Graph
      3-coloring},label={fig:threecol}}
\begin{lstlisting}
e(X,Y) :- e(Y,X).
n(X) :- e(X,_).
color(1..3).
{c(X,Y) : color(Y)} = 1 :- n(X).
:- e(X,Y), c(X,C), c(Y,D), not C <> D.
\end{lstlisting}
\lstset{frame=single,numbers=left,caption={Minimal $k$-coloring (extending Listing~\ref{fig:threecol})},label={fig:minkcol}}
\begin{lstlisting}
color(X) :- n(X).
cost(C,1) :- c(_,C).
#minimize { Cost,C : cost(C,Cost) }.
\end{lstlisting}
\end{figure}
Figure~\ref{fig:coloring} shows an example graph where edge
relationships correspond to land borders between some countries.  The
edges are defined using an association list notation; for example
\lstinline{e(be,(lu;nl))} abbreviates two edges \lstinline{e(be,lu)}
and \lstinline{e(be,nl)}.  Listing~\ref{fig:threecol} defines graph
3-coloring declaratively.  The first
line states that the edge relation is symmetric and the second defines
the node relation to consist of all sources (and by symmetry targets)
of edges.  Line 3 defines a relation \lstinline{color/1} to hold for
values 1,2,3.
Lines 4--5 define when a graph is 3-colorable, by defining when a
relation \lstinline{c/2} is a valid 3-coloring.  Line 4 says that
\lstinline{c/2} represents a (total) function from nodes to colors,
i.e. for every node there is exactly one associated color.  Line 5
says that for each edge, the associated colors of the source and
target must be different.  Here, we are using the \lstinline{not}
operator solely to illustrate its use, but we could have done without
it, writing \lstinline{C = D} instead.

Listing~\ref{fig:threecol} is a complete
program that can be used with 
Figure~\ref{fig:coloring} to determine that the example graph is not
3-colorable. 
What if we want to find the least $k$ such that a graph is
$k$-colorable?  We cannot leave the number of colors undefined, since
ASP requires a finite search space, but we could manually change the `3' on line 5
to various values of $k$, starting with the maximum $k = |V|$ and
decreasing until the minimum possible $k$ is found.


Instead, using \emph{minimization constraints}, we can modify the
3-coloring program above to instead compute a minimal $k$-coloring
(that is, find a coloring minimizing the number of colors) purely
declaratively by adding the clauses shown in
Listing~\ref{fig:minkcol}. Line 1 defines the set of colors simply to
be the set of node identifiers (plus the three colors we already had,
but this is harmless).  Line 2 associates a cost of 1 with each used
color.  Finally, line 3 imposes a minimization constraint:to minimize
the sum of the costs of the colors.
Thus, using a single Clingo specification we can automatically find
the minimum number of colors needed for this (or any) undirected
graph.  The 4-coloring shown in Figure~\ref{fig:coloring} was found
this way.

\section{Specifying graph matching and edit distance}

In this section we give ASP specifications defining each problem.  We
first consider how to represent graphs as flat collections of facts,
suitable for use in a logic programming setting.  We choose one among
several reasonable representations: given $G = (V,E,src,tgt,lab,prop)$
and given three predicate names $\mathtt{n},\mathtt{e},\mathtt{p}$ we define the following
relations:
\begin{eqnarray*}
  Rel_G(\mathtt{n},\mathtt{e},\mathtt{p}) &=& \{\mathtt{n}(v,lab(v)) \mid v \in V\}\\
&\cup& \{\mathtt{e}(e,src(e),tgt(e),lab(e)) \mid e \in E\}\\
&\cup& \{\mathtt{p}(x,k,d) \mid x \in V\cup E, prop(x,k) = d \neq \bot\}
\end{eqnarray*}
Clearly, we can recover the original graph from this representation. 

In the following problem specifications, we always consider two
graphs, say $G_1$ and $G_2$, and to avoid confusion between them we
use two sets of relation names to encode them, thus
$Rel_{G_1}(\mathtt{n}_1,\mathtt{e}_1,\mathtt{p}_1) \cup
Rel_{G_2}(\mathtt{n}_2,\mathtt{e}_2,\mathtt{p}_2)$
represents two graphs.  We also assume without loss of generality that
the sets of vertex and edge identifiers of the two graphs are all
disjoint, i.e.  $(V_1 \cup E_1) \cap (V_2 \cup E_2) = \emptyset$,
to avoid any possibility of confusion among them.



We now show how to specify homomorphisms and isomorphisms among
graphs.  The Clingo code in Listing~\ref{fig:gh} defines when a graph
homomorphism exists from $G_1$ to $G_2$.  We refer to this program
extended with suitable representations of $G_1$ and $G_2$ as
$Hom_h(G_1,G_2)$. The binary relation $h$, representing the
homomorphism, is specified using two constraints.  The first says that
$h$ maps nodes of $G_1$ to nodes of $G_2$ with the same label, while
the second additionally specifies that $h$ maps edges of $G_1$ to
those of $G_2$ preserving source, target, and label.  Notice in
particular that the cardinality constraint ensures that $h$ represents
a total function with range $V_1 \cup E_1$, so in any model satisfying
the first clause, every node in $G_1$ is matched to one in $G_2$,
which means that the body of the second clause is satisfiable for each
edge.  The third clause simply constrains $h$ so that any properties
of nodes or edges in $G_1$ must be present on the matching node or
edge in $G_2$.

\begin{figure}[t]
\lstset{frame=single,numbers=left,caption={Graph homomorphism},label={fig:gh}}
\begin{lstlisting}
{h(X,Y) : n2(Y,L)} = 1 :- n1(X,L).
{h(X,Y) : e2(Y,S2,T2,L), h(S1,S2), h(T1,T2)} = 1 :- e1(X,S1,T1,L).
:- p1(X,K,D), h(X,Y), not p2(Y,K,D).
\end{lstlisting}
\lstset{frame=single,numbers=left,caption={Graph isomorphism
    (extending Listing~\ref{fig:gh})},label={fig:gi}}
\begin{lstlisting}
{h(X,Y) : n1(X,L)} = 1 :- n2(Y,L).
{h(X,Y) : e1(X,S1,T1,L), h(S1,S2), h(T1,T2)} = 1 :- e2(Y,S2,T2,L).
:- p2(Y,K,D), h(X,Y), not p1(X,K,D).
\end{lstlisting}
\lstset{frame=single,numbers=left,caption={Subgraph isomorphism (extending Listing~\ref{fig:gh})},label={fig:sgi}}
\begin{lstlisting}
{h(X,Y) : n1(X,L)} <= 1 :- n2(Y,L).
{h(X,Y) : e1(X,S1,T1,L), h(S1,S2), h(T1,T2)} <= 1 :- e2(Y,S2,T2,L).
\end{lstlisting}
\end{figure}
Next to define when $h$ is a graph \emph{isomorphism}, we add the
symmetric clauses shown in Listing~\ref{fig:gi}.  We write
$Iso_h(G_1,G_2)$ for the combination of Listings~\ref{fig:gh}
and~\ref{fig:gi}.  Since the two  listings together imply that $h$
represents a homomorphism in the forward direction and simultaneously
represents a homomorphism from $G_2$ to $G_1$ in the backward
direction, these four clauses suffice to specify that $h$ is an isomorphism.

To specify subgraph isomorphism, we simply require that $h$ is an
injective homomorphism from $G_1$ to $G_2$, as shown in
Listing~\ref{fig:sgi}.  We refer to the specification in
Listing~\ref{fig:sgi} as $Sub_h(G_1,G_2)$.  The two additional
constraints specify that the inverse of $h$ is a \emph{partial}
homomorphism.  This is equivalent to $h$ being an injective homomorphism.

Finally we consider the specification of the graph edit distance
problem.  On the surface, this seems challenging, since the graph edit
distance is defined as the length  of a minimal edit script mapping one
graph to another, and there are infinitely many possible edit
scripts.  However, there is clearly always an upper bound $d$ on the
edit distance: consider an edit script that just deletes
$G_1$ and inserts $G_2$, and
take $d$ to be the length of this script.
So, given two graphs and this upper bound $d$ we could proceed by
specifying a search space over edit scripts of bounded length,
defining the meaning of each edit operator, and seeking to minimize
the number of steps necessary to get from $G_1$ to $G_2$.  However,
this encoding seems rather heavyweight, and requires 
preprocessing to determine $d$.

Instead, we follow a different strategy, analogous to the approach
adopted for graph coloring earlier.  The strategy is based on the
observation that the graph edit distance is closely related to the
\emph{maximum subgraph problem}~\cite{bunke97}, that is, given two
graphs $G_1$, $G_2$, find the largest graph that is subgraph
isomorphic to both.  If we identify such a graph then (as we shall
show) we can read off an edit script that maps $G_1$ to $G_2$, which
first deletes unmatched structure from $G_1$, then updates
properties in-place, and finally inserts new structure needed in
$G_2$.  Furthermore, to identify the maximum common subgraph, we do not need
to construct a new graph separate from $G_1$ and $G_2$; instead, we
can think of the maximum common subgraph as an isomorphic pair of
subgraphs of $G_1$ and $G_2$.  So in other words, we will search for a
partial isomorphism $h$ between $G_1$ and $G_2$, use it as a basis for
extracting an edit script, and minimize its cost.

\begin{figure}[tb!]
\lstset{frame=single,numbers=left,caption={Graph edit distance},label={fig:ged}}
\begin{lstlisting}
{h(X,Y) : n2(Y,L)} <= 1 :- n1(X,L).
{h(X,Y) : n1(X,L)} <= 1 :- n2(Y,L).
{h(X,Y) : e2(Y,S2,T2,L), h(S1,S2), h(T1,T2)} <= 1 :- e1(X,S1,T1,L).
{h(X,Y) : e1(X,S1,T1,L), h(S1,S2), h(T1,T2)} <= 1 :- e2(Y,S2,T2,L).

delete_node(X) :- n1(X,_), not h(X,_).
insert_node(Y,L) :- n2(Y,L), not h(_,Y).

delete_edge(X) :- e1(X,_,_,_), not h(X,_).
insert_edge(Y,S,T,L) :- e2(Y,S,T,L), not h(_,Y).

update_prop(X,K,V1,V2) :- p1(X,K,V1), h(X,Y), p2(Y,K,V2), V1 <> V2.
delete_prop(X,K) :- p1(X,K,_), h(X,Y), not p2(Y,K,_).
delete_prop(X,K) :- p1(X,K,_), delete_node(X).
delete_prop(X,K) :- p1(X,K,_), delete_edge(X).
insert_prop(Y,K,V) :- p2(Y,K,V), h(X,Y), not p1(X,K,_).
insert_prop(Y,K,V) :- p2(Y,K,V), insert_node(Y,_).
insert_prop(Y,K,V) :- p2(Y,K,V), insert_edge(Y,_,_,_).

node_cost(Y,1) :- insert_node(Y,_).
node_cost(X,1) :- delete_node(X).

edge_cost(Y,1) :- insert_edge(Y,_,_,_).
edge_cost(X,1) :- delete_edge(X).

prop_cost(X,K,1) :- update_prop(X,K,V1,V2).
prop_cost(X,K,1) :- delete_prop(X,K).
prop_cost(Y,K,1) :- insert_prop(Y,K,V).

#minimize { NC,X : node_cost(X,NC);
              EC,X : edge_cost(X,EC);
              LC,X,K : prop_cost(X,K,LC)}.
\end{lstlisting}
\end{figure}

Listing~\ref{fig:ged} accomplishes this.  The first four lines specify
that $h$ must be a partial isomorphism, by dropping the
requirement that $h$ must match all nodes/edges on one side with those
of another, and dropping the hard constraint that properties must match.
Lines 6--7 define when a node must be deleted or inserted.  Nodes that
are in $G_1$ and not matched in $G_2$ must be deleted, and conversely
those that are in $G_2$ and not matched in $G_1$ must be inserted.
Lines 9--10 similarly specify when edges must be inserted or deleted.
Lines 12--18 define when a property is updated in-place, deleted, or
inserted.  If a property key is present on an object in $G_1$ and on
the matching object in $G_2$ but with a different value, then the
key's value needs to be updated.  If it is present in $G_1$ but not
present on the matching object in $G_1$ then it is deleted.
Likewise, if it is present in $G_1$ but the associated object is
deleted then the property also must be deleted. 
Dually, properties are inserted if they are present in $G_2$ but not
in $G_1$, either because the matching object does not have that
property or because there is no matching object because the property
is on an inserted object.
Lines 20--28 specify the costs associated with each of the edit
operations.  We assign each operation a cost of 1.  It would also be
possible to assign different (integer) costs to different kinds of updates, or
even to specify different costs depending on labels, keys, or values.

\section{Correctness}

We first state the intended correctness properties for the
homomorphism, isomorphism, and subgraph isomorphism problems:
\begin{theorem}\label{thm:correctness}
  \begin{enumerate}
  \item There exists a homomorphism $h : G_1 \to G_2$ if and only if
    $ Hom_h(G_1,G_2)$ is satisfiable.
  \item There exists an isomorphism $h : G_1 \to G_2$ if and only if
    $ Iso_h(G_1,G_2)$ is satisfiable.  
  \item  $h : G_1 \to G_2$ witnesses a subgraph isomorphism if and only if
    $Sub_h(G_1,G_2)$ is satisfiable.  
\end{enumerate}
\end{theorem}
\begin{proof}
See \refapp{app:proofs}.
  \qed
\end{proof}

\begin{figure}[t!]
\[\small
\begin{array}{rcl}
  \delE^*(e);\delP(x,k) &\longrightarrow & \delP(x,k);\delE^*(e)
\smallskip\\
  \delV^*(v);\delP(x,k) &\longrightarrow & \delP(x,k);\delV^*(v)
\\
  \delV^*(v);\delE(e) &\longrightarrow & \delE(e);\delV^*(v)
\smallskip\\
  \updP^*(x,k,d);\delP(y,k') &\longrightarrow & \left\{
                                              \begin{array}{ll}
                                                \delP(y,k')& \text{if }x=y,k=k'\\
                                                \delP(y,k');
                                                \updP^*(x,k,d) & \text{otherwise}
                                              \end{array}
\right.\\
  \updP^*(x,k,d);\delE(e) &\longrightarrow & \delE(e);\updP^*(x,k,d)
\\
  \updP^*(x,k,d);\delV(v) &\longrightarrow & \delV(v);\updP^*(x,k,d)
\smallskip\\
\insV^*(v,l); \delP(x,k) &\longrightarrow& \delP(x,k);\insV^*(v,l)
\\
\insV^*(v,l); \delE(e) &\longrightarrow& \delE(e);\insV^*(v)
\\
\insV^*(v,l); \delV(v') &\longrightarrow& \left\{
                                              \begin{array}{ll}
                                                \epsilon& \text{if }v=v'\\
                                                \delV(v');
                                                \insV^*(v,l) & \text{otherwise}
                                              \end{array}\right.
\\
\insV^*(v,l); \updP(x,k,d) &\longrightarrow& \updP(x,k,d);\insV^*(v,l)
\smallskip\\
\insE^*(e,v,w,l); \delP(x,k) &\longrightarrow& \delP(x,k);\insE^*(e,v,w,l)
\\
\insE^*(e,v,w,l); \delE(e') &\longrightarrow& \left\{
                                              \begin{array}{ll}
                                                \epsilon& \text{if }e=e'\\
                                                \delE(e');
                                                \insE^*(e,v,w,l) & \text{otherwise}
                                              \end{array}\right.
\\
\insE^*(e,v,w,l); \delV(v') &\longrightarrow& \delV(v'); \insE^*(e,v,w,l)
\\
\insE^*(e,v,w,l); \updP(x,k,d) &\longrightarrow& \updP(x,k,d); \insE^*(e,v,w,l)
\\
\insE^*(e,v,w,l); \insV(v',l) &\longrightarrow& \insV(v',l);
                                                \insE^*(e,v,w,l)
\smallskip\\
\insP^*(x,k,d); \delP(y,k') &\longrightarrow& \left\{
                                              \begin{array}{ll}
                                                \epsilon& \text{if }x=y,k=k'\\
                                                \delP(y,k');
                                                \insP^*(x,k,d) & \text{otherwise}
                                              \end{array}
\right.\\
\insP^*(x,k,d); \delE(e) &\longrightarrow& 
                                                \delE(e);
                                                \insP^*(x,k,d) 
\\
\insP^*(x,k,d); \delV(v) &\longrightarrow& 
                                                \delV(v);
                                                \insP^*(x,k,d) 
\\
\insP^*(x,k,d); \updP(y,k',d') &\longrightarrow& \left\{
                                                \begin{array}{ll}
                                                \insP^*(x,y,d')& \text{if }x=y,k=k'\\
                                                  \updP(y,k',d');
                                                \insP^*(x,k,d) & \text{otherwise}
                                              \end{array}
\right.\\
\insP^*(x,k,d); \insV(v,l) &\longrightarrow& 
                                                \insV(v',l);
                                                \insP^*(x,k,d) \\
\insP^*(x,k,d); \insE(e,v,w,l) &\longrightarrow& 
                                                \insE(e,v,w,l);
                                                \insP^*(x,k,d) \\
op^*; ops &\longrightarrow& 
                                                op; ops \qquad
                            \text{if no earlier rule applies}
\end{array}
\]\caption{Edit script rewrite rules}\label{fig:rewrite}
\end{figure}

Next we turn to graph edit distance.  
To assist with the reasoning, we define the following
canonical form:
\begin{definition}
  [Edit script canonical form]
An edit script is in \emph{canonical form} if it is of the form
$del_{p};del_{e};del_{v};upd_p;ins_v;ins_e;ins_p$, 
where:
\begin{itemize}
\item $del_p$, $del_e$ and $del_v$ are sequences of property
  deletions, edge deletions, and node deletions respectively;
\item $upd_p$ is a sequence of property updates;
\item $ins_v$,  $ins_e$, and $ins_p$ are sequences of node insertions,
  edge insertions, and property insertions, respectively.
\end{itemize}
\end{definition}
Edit scripts obtained from $GED_h(G_1,G_2)$ are in this form.
Moreover, any valid edit script can be converted to a canonical one by
applying a set of rewrite rules, as shown in Figure~\ref{fig:rewrite}.
We first consider \emph{marked} versions $op^*$ of each edit
operation, for example writing $\delP^*(x,k)$ for the marked version
of $\delP$.  A marked operation $op^*$ has the same effect as $op$
when applied to a graphs; the mark is only to indicate which operation
is actively being rewritten.  The idea here is that if we have a
canonical edit script $ops$ and wish to add a new edit operation, we
use the rewrite rules to canonicalize $op^*;ops$.  The rules are
applied in order and at each step, the first matching rule is applied.
Note that there is a catch-all rule $op^*;ops \longrightarrow op;ops$,
which only applies if none of the other rules do. Essentially, the
rewrite rules consider all of the possible pairs of adjacent
operations that can appear in a non-canonical form, with the first
element marked.  In each case, they show how to simplify the edit
script by either moving the marked operation closer to the end, or
removing the mark.  Removal can happen as a result of either
cancellation of the marked operation by another operation (e.g. a
delete undoing an insert), or by removing the mark once it has reached
an appropriate place for it in the canonical form.

\begin{lemma}
  \label{lem:canonical}
 If $ops$ is an edit script mapping $G_1$ to $G_2$, then there is
  a canonical edit script $ops'$ mapping $G_1$ to $G_2$ such that
  $|ops'| \leq |ops|$.
\end{lemma}
\begin{proof}
See \refapp{app:proofs}.
\qed
\end{proof}

\begin{theorem}
  The specification $GED_h(G_1,G_2)$ always has a solution, and the
  edit script described by the insertion, deletion and update
  predicates is a valid, canonical script mapping $G_1$ to $G_2$.
  Moreover, the 
  cost of the optimal solution to $GED_h(G_1,G_2)$ equals $GED(G_1,G_2)$.
\end{theorem}
\begin{proof}
For the first part, we observe that the empty relation $ h =
\emptyset$ always solves $GED_h(G_1,G_2)$ if we ignore the
minimization constraint.  Therefore, the cost of this solution is an
upper bound.  Moreover, if we apply the edit
operations described by the insert, delete and update relations in the
order required by the canonical form, 
then each edit operation is valid, all structure present in $G_1$ and
not $G_2$ is removed, all properties whose values differ in $G_1$ and
$G_2$ are updated, and all structure present in $G_2$ and not $G_1$ is
inserted.  Therefore, the corresponding edit script maps $G_1$ to $G_2$.

To show that the minimum cost obtained from solving the
$GED_h(G_1,G_2)$ specification coincides with $GED(G_1,G_2)$, one
direction is easy: for any $h$ (including the one corresponding to a
minimum cost solution) the collection of edit operations resulting
from $GED_h(G_1,G_2)$ is a valid edit script so its length $d$ must be
greater than or equal to the minimum over all valid scripts.  To show
the reverse direction, we use Lemma~\ref{lem:canonical}.  Given a
minimum-length edit script that is not in canonical form, we can
rewrite it to one that is canonical, with equal cost (since the original script
was already minimum-length).  \qed
\end{proof}

\section{Discussion}
We have argued that using ASP offers considerable flexibility.  To illustrate this
claim, we consider three modifications to our approach.

\paragraph{Weighted graph edit distance} If the operations have
different (integer) weights, implemented using a suitable modification
to the cost predicates in some specification $wGED_h(G_1,G_2)$, then
the same argument as above suffices to show that a minimum-weight canonical
script always exists to be found by the ASP specification.  The key
point is that weights are defined on individual edit operations, and
the rewrite rules only permute or delete operations, so preserve or
decrease weight.  

\paragraph{Relabeling} We have treated labels as hard constraints: it
is not possible to change the label of a node in $G_1$ to a different
label in $G_2$, short of deleting the node and inserting a new one
with a different label.  On the other hand, properties are soft
constraints in the sense that we may delete or update a property value
without also being obliged to delete and re-create the underlying node
or edge structure.  It is natural to consider an in-place relabeling
operation as well.  Such behavior can be encoded on top of the
already-developed framework by using a single ``blank'' label for
nodes and edges and introducing an unused property key called
``label'' instead; now this can be updated in-place like other
properties.  Alternatively, we can accommodate this behavior 
more directly as shown in Listing~\ref{fig:ged-relabel}.
\begin{figure}[tb]
  \lstset{frame=single,numbers=left,caption=Graph edit distance with
    relabeling (modifies Listing~\ref{fig:ged}),label={fig:ged-relabel}}
\begin{lstlisting} 
{h(X,Y) : n2(Y,_)} <= 1 :- n1(X,_).  
{h(X,Y) : n1(X,_)} <= 1 :- n2(Y,_).  
{h(X,Y) : e2(Y,S2,T2,_), h(S1,S2), h(T1,T2)} <= 1 :- e1(X,S1,T1,_).  
{h(X,Y) : e1(X,S1,T1,_), h(S1,S2), h(T1,T2)} <= 1 :- e2(Y,S2,T2,_).
...  
relabel_node(X,L2) :- n1(X,L1),h(X,Y), n2(Y,L2), L1 <> L2.  
relabel_edge(X,L2) :- e1(X,_,_,L1),h(X,Y),e2(Y,_,_,L2), L1 <> L2.  
...  
node_cost(X,1) :- relabel_node(X,_).  
edge_cost(X,1) :- relabel_edge(X,_).
\end{lstlisting}
\end{figure}
The first four lines relax the constraint that node and edge labels
have to be preserved by $h$.  The next two lines define the
\lstinline{relabel_node} and \lstinline{relabel_edge} predicates to
detect when two matched nodes or edges have different labels.
Finally, the \lstinline{node_cost} and \lstinline{edge_cost}
predicates are extended to charge a cost of 1 per relabeling.

\paragraph{Ad hoc constraints}
The use of ASP opens up many other possibilities for controlling or
constraining the various isomorphism or edit distance problems.
One example which we found useful in previous work~\cite{chan19mw} was to modify the
definitions of isomorphism or subgraph isomorphism to treat properties
as soft constraints and minimize the number of mismatched properties.

Another potentially interesting class of constraints is to allow
``access control'' constraints on the possible edit scripts, for
example specifying that certain nodes or edges in one graph cannot not
be modified and so must be matched with equivalent constructs in the
other graph.  This is similar to the approximate constrained subgraph
matching problem~\cite{zampelli05cp}.

\section{Evaluation}

Graph matching and edit distance are widely studied problems and a thorough
comparison of our approach with state-of-the-art algorithms is beyond
the scope of this paper.  However, we do not claim that our approach
is faster, only that it is easy to implement and modify, rendering it
suitable for rapid prototyping situations.  Nevertheless, in this
section we summarize a preliminary evaluation that supports a claim
that our approach is fast enough to be useful for rapid prototyping.  Our
experiments were run on an 2.6 GHz Intel Core i7 MacBook Pro machine
with 8 GB RAM and using Clingo v5.2.0.

First, we consider the various problems on synthetic graphs, such as
$k$-cycles and $k$-chains (linear sequences of $k$ edges), with only
one possible node and edge label and no properties.  These problems
are not representative of typical real problems, but illustrate some
general trends.  We considered each of the problems: (HOM), (ISO)
$G_1 \cong G_2$, (SUB) $S_n \lesssim C_n$, and (GED) $GED(G_1,G_2)$.
We first considered comparisons where $G_1$ and $G_2$ are $k$-cycles
or $k$-chains, for $k \in \{10,20,\ldots,100\}$.  We found the running
times for each of these problems to be relatively stable independent
of whether the comparison was between two $k$-chains, a
$k$-chain with a $k$-cycle, or two $k$-cycles, so we have averaged
across all four scenarios.  We also considered randomly generated
graphs with $k$ nodes and each edge generated with probability 0.1,
with $k \in \{5,10,\ldots,50\}$.  We attempted each problem with a
running time limit of 30 seconds; the results are shown in
Figure~\ref{fig:synthresults} results.  Unsurprisingly, the HOM
instances are solved fastest, and GED slowest.

\begin{figure}[tb]
  \centering
  \begin{tabular}{cc}
(a) & (b)\\
  \includegraphics[scale=0.5]{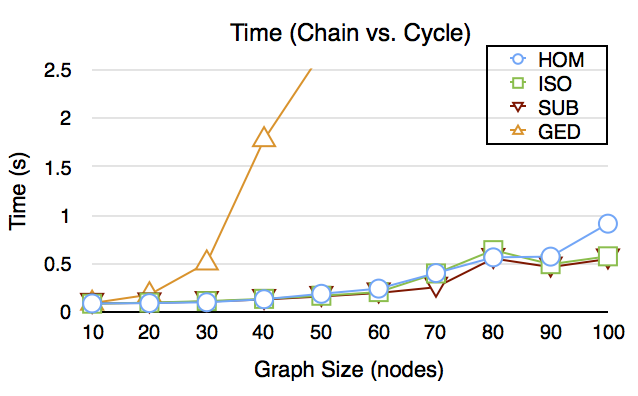}
&
\includegraphics[scale=0.5]{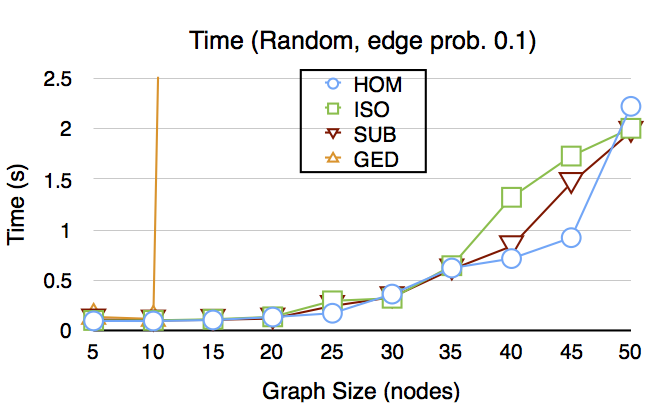}
  \end{tabular}\caption{Synthetic results: (a) chains and cycles (b)
    randomly generated graphs}
  \label{fig:synthresults}
\end{figure}

Second, we consider some real graphs from the Mutagenesis dataset
(MUTA), a standard dataset used for evaluating graph edit distance
algorithms~\cite{muta}, for example in a recent graph edit distance
competition (GEDC)~\cite{abuaisheh17prl}.  In the contest, eight
algorithms were run on different problems for up to 30 seconds, and
compared in terms of time, accuracy (for approximate algorithms), and
success rate (for exact algorithms).  We modified the GED specification to
allow node and edge relabeling and use the same weight function as in
the second (and more challenging) configuration used in the contest,
for which even the best algorithm (called F2) was not able to deal
with graphs of size larger than 30.  We consider three datasets
MUTA-10, MUTA-20 and MUTA-30 each consisting of ten chemical structure
graphs of size 10, 20 or 30 respectively.  We also consider a dataset
MUTA-MIXED which consists of ten graphs of varying sizes.  We
considered all unordered pairs of the graphs in each subset and attempted to find
the GED with a timeout of 30s.  Table~\ref{tab:gedc} shows the results
compared with the four exact algorithms reported
in~\cite{abuaisheh17prl}. The first two algorithms, F2 and F24threads,
are implementations of a binary linear programming encoding of graph
edit distance~\cite{lerouge17pr}, the first being the plain
single-threaded algorithm, the second a variant using a linear
programming solver, and the second running with four threads.  The
other two, DF and PDFS, are sequential and parallel implementations of
a depth-first, branch-and-bound algorithm~\cite{abuaisheh18esa,abuaisheh15icpram}.

Table~\ref{tab:gedc} illustrates that our approach is competitive
with DF and slightly worse than PDFS, but does not match the performance of the two
F2 algorithms.  These results should be taken with a grain of salt,
since we have not replicated the GEDC results on our (slightly faster)
hardware.  Memory did not appear to be a bottleneck for our approach.

\begin{table}[tb]
\caption{Success rate (optimal solution found in under 30s) on Mutagenesis dataset}\label{tab:gedc}
\begin{center}
  \begin{tabular}[tb]{|l|cccc|}
    \hline
    & MUTA-10 & MUTA-20 & MUTA-30 & MUTA-MIXED\\
    \hline
    F24threads~\cite{abuaisheh17prl,lerouge17pr}$^\dagger$ & 100\% &
                                                                     98\%
                        & 23\% & 44\%\\
    F2~\cite{abuaisheh17prl,lerouge17pr}$^\dagger$ & 100\% & 94\% &
                                                                    15\%
                                  & 41\%\\
 PDFS~\cite{abuaisheh17prl,abuaisheh18esa}$^\dagger$ & 100\% & 26\% & 11\% & 10\%\\
    Our approach & 100\% & 26\% & 10\%  & 4\%\\
    DF~\cite{abuaisheh17prl,abuaisheh15icpram}$^\dagger$ & 100\% & 14\% & 10\% & 10\% \\
    \hline
  \end{tabular}
\end{center}
\noindent \footnotesize
$^\dagger$Experiments from~\cite{abuaisheh17prl} run on a
4-core 2.0GHz AMD Opteron 8350 with 16GB RAM.
\end{table}

We have implemented and used variations of the isomorphism and
subgraph isomorphism specifications for property graphs in a
provenance graph analysis system called ProvMark~\cite{chan19mw}.  In
this earlier work, we found that for
graphs of up to around 100 nodes and edges, and a few hundred
properties, these problems are usually solvable within a few seconds.  
However, these problems may not be representative of other scenarios.

The specifications we used to define approximate subgraph isomorphism
problems in ProvMark are similar to those presented here, but we
subsequently experimented with several different approaches with
different performance.  Here, we compare the performance of ProvMark
on subgraph isomorphism problems over two representative example
graphs considered in our previous experiments: the graph
generalization and comparison problems resulting from benchmarking the
\texttt{creat} and \texttt{execve} system calls using the CamFlow
provenance recording system~\cite{pasquier17socc}.
See~\cite{chan19mw} for further details and the Clingo code of the
previous approaches.  

\begin{table}[tb]\label{tab:provmark}
\caption{Performance improvement vs.\ Provmark~\cite{chan19mw}}
\centering
\begin{tabular}{|l|c|ccc|}
\hline
  Experiment & Size
& Old Time (s)& New Time (s) & Speedup\\
\hline
creat-bg-gen & 1006 &0.060  & 0.034  &  1.9x 
\\
creat-fg-gen &  1060 & 0.070 & 0.037 & 1.9x
\\
creat-comp & 1033 & 0.053 & 0.026 & 2.1x\\
execve-bg-gen & 1006 & 0.061 & 0.036 & 1.7x
\\
execve-fg-gen & 1340 & 0.114 & 0.051 & 2.2x
\\
execve-comp & 1173 & 0.083 & 0.042 & 1.9x
\\
\hline
\end{tabular}
\end{table}
Table~\ref{tab:provmark} shows the running time of the old version and
new version of approximate subgraph isomorphism. The code for both
specifications is in \refapp{app:code}.  The problem sizes
(that is, the number of nodes, edges, and properties of the two
graphs) is shown under ``Size''.  The ``Old Time'' column corresponds
to the time obtained using the old approach and ``New Time'' shows the
time obtained using the code in Listing~\ref{fig:sgi} modified to
allow approximate property matching.  The ``Speedup'' column shows the
ratio between the old and new time.  In most cases, the speedup is
around a factor of two.  As future work, we plan to use graph
edit distance with the results of the ProvMark system, for example for
clustering or regression testing across runs.

\section{Related Work}

The lower bound of the complexity of graph isomorphism is a well-known
open problem~\cite{arvind05beatcs}, but subgraph isomorphism and graph
edit distance are NP-complete~\cite{garey79}. A number of practical
algorithms for graph isomorphism have been studied, however, including
NAUTY~\cite{mckay81cn}, which has also been integrated with
Prolog~\cite{frank16tplp}. However, most such algorithms consider
graphs with vertex labels but not edge labels or properties, so are
not directly applicable to property graph isomorphism.  Subgraph
isomorphism has been studied extensively over the past years, one
survey~\cite{mckay2014jsc} summarizes the state-of-art algorithms for
solving partial or simplified version of the problem. Subgraph
isomorphism is also studied for graph databases, where the query
subgraph is usually small but the other graph may be very large.  Lee
et al.~\cite{lee12vldb} evaluated five such algorithms on query graphs
of up to 24 edges and database of up to tens of thousands of nodes and
edges.  Approximate subgraph matching with constraints has also been
studied, particularly in biomedical settings~\cite{zampelli05cp}, and
it would be interesting to investigate whether our approach is
competitive with their CSP-based algorithm.  Graph edit distance has
also been studied extensively~\cite{gao2010paa}, with much attention
on approximate algorithms that can provide results
quickly~\cite{riesen15}.

While several approaches to graph matching and edit distance have been
based on expressing these problems as constraint satisfaction
problems, satisfiability, or linear programming problems, to the best
of our knowledge there is no previous work based on answer set
programming.  Moreover, our approach easily accommodates richer graph
structure such as hard or soft label constraints, properties, and
multiple edges between pairs of nodes, whereas the algorithms we have
seen generally consider ordinary  graphs (without properties and with at
most one edge between two nodes).

\section{Conclusions}

The graph edit distance problem is a widely studied problem that has
many applications.  Exact solutions to it, and to related problems
such as graph isomorphism and subgraph isomorphism, are challenging to
compute efficiently due to their NP-completeness or unresolved
complexity (in the case of graph isomorphism).  There are a number of
proposed algorithms in the literature, with one of the most effective
based on a reduction to binary linear programming~\cite{lerouge17pr}.  In this
paper, we investigated an alternative approach using answer set
programming (ASP), specifically the Clingo solver.  This approach may
not be competitive with the best known techniques in terms of
performance, but has the potential advantage that it is
straightforward to modify the problem specification to accommodate
different kinds of graphs, cost metrics or other variations, or to
accommodate ad hoc constraints that can also be expressed using ASP.
Our approach has already proved useful for a real
application~\cite{chan19mw}, and our experimental evaluation suggests
that it is also competitive with two out of four exact algorithms from
a
graph edit distance competition.

Our work may be valuable to others interested in rapid prototyping of
graph matching or edit distance problems using declarative
programming.  Additional work could be done to facilitate this, for
example using Clingo's Python wrapper library.  Graph matching and
edit distance problems may also be an interesting class of challenge
problems for developers of ASP solvers.

\section*{Acknowledgments}
Effort sponsored by the Air Force Office of Scientific Research, Air
Force Material Command, USAF, under grant number
\grantnum{AFOSR}{FA8655-13-1-3006}. The U.S. Government and University
of Edinburgh are authorised to reproduce and distribute reprints for
their purposes notwithstanding any copyright notation thereon.  Cheney
was also supported by ERC Consolidator Grant Skye (grant number
\grantnum{ERC}{682315}).  This material is based upon work supported
by the Defense Advanced Research Projects Agency (DARPA) under
contract \grantnum{DARPA}{FA8650-15-C-7557}.

\bibliographystyle{plain}
\bibliography{paper}

\begin{thebibliography}{10}

\bibitem{abuaisheh17prl}
Zeina Abu{-}Aisheh, Benoit Ga{\"{u}}z{\`{e}}re, S{\'{e}}bastien Bougleux,
  Jean{-}Yves Ramel, Luc Brun, Romain Raveaux, Pierre H{\'{e}}roux, and
  S{\'{e}}bastien Adam.
\newblock Graph edit distance contest: Results and future challenges.
\newblock {\em Pattern Recognition Letters}, 100:96--103, 2017.

\bibitem{abuaisheh15icpram}
Zeina Abu{-}Aisheh, Romain Raveaux, Jean{-}Yves Ramel, and Patrick Martineau.
\newblock An exact graph edit distance algorithm for solving pattern
  recognition problems.
\newblock In {\em Proceedings of the International Conference on Pattern
  Recognition Applications and Methods ({ICPRAM} 2015)}, pages 271--278, 2015.

\bibitem{abuaisheh18esa}
Zeina Abu{-}Aisheh, Romain Raveaux, Jean{-}Yves Ramel, and Patrick Martineau.
\newblock A parallel graph edit distance algorithm.
\newblock {\em Expert Syst. Appl.}, 94:41--57, 2018.

\bibitem{arvind05beatcs}
Vikraman Arvind and Jacobo Tor\'an.
\newblock Isomorphism testing: Perspectives and open problems.
\newblock {\em Bulletin of the {EATCS}}, 86:66–84, 2005.

\bibitem{dbpedia}
S{\"{o}}ren Auer, Christian Bizer, Georgi Kobilarov, Jens Lehmann, Richard
  Cyganiak, and Zachary~G. Ives.
\newblock {DBpedia}: {A} nucleus for a web of open data.
\newblock In {\em Proceedings of the 6th International Semantic Web Conference
  on The Semantic Web and 2nd Asian Semantic Web Conference ({ISWC} 2007 +
  {ASWC} 2007)}, pages 722--735, 2007.

\bibitem{bunke97}
Horst Bunke.
\newblock On a relation between graph edit distance and maximum common
  subgraph.
\newblock {\em Pattern Recognition Letters}, 18(8):689--694, 1997.

\bibitem{chan19mw}
{Sheung Chi} Chan, James Cheney, Pramod Bhatotia, Thomas Pasquier, Ashish
  Gehani, Hassaan Irshad, Lucian Carata, and Margo Seltzer.
\newblock {ProvMark}: A provenance expressiveness benchmarking system.
\newblock In {\em Proceedings of the 20th International Middleware Conference
  (Middlware '19)}, pages 268--279. ACM, 2019.

\bibitem{chen19kbs}
Xiaoyang Chen, Hongwei Huo, Jun Huan, and Jeffrey~Scott Vitter.
\newblock An efficient algorithm for graph edit distance computation.
\newblock {\em Knowl.-Based Syst.}, 163:762--775, 2019.

\bibitem{frank16tplp}
Michael Frank and Michael Codish.
\newblock Logic programming with graph automorphism: Integrating nauty with
  prolog (tool description).
\newblock {\em {TPLP}}, 16(5-6):688--702, 2016.

\bibitem{gao2010paa}
Xinbo Gao, Bing Xiao, Dacheng Tao, and Xuelong Li.
\newblock A survey of graph edit distance.
\newblock {\em Pattern Analysis and applications}, 13(1):113--129, 2010.

\bibitem{garey79}
Michael~R. Garey and David~S. Johnson.
\newblock {\em Computers and intractability: a guide to the theory of
  NP-completeness}.
\newblock W. H. Freeman, 1979.

\bibitem{gebser2018ki}
Martin Gebser, Roland Kaminski, Benjamin Kaufmann, Patrick L{\"u}hne, Philipp
  Obermeier, Max Ostrowski, Javier Romero, Torsten Schaub, Sebastian
  Schellhorn, and Philipp Wanko.
\newblock The potsdam answer set solving collection 5.0.
\newblock {\em KI-K{\"u}nstliche Intelligenz}, 32(2-3):181--182, 2018.

\bibitem{gebser2011aicom}
Martin Gebser, Benjamin Kaufmann, Roland Kaminski, Max Ostrowski, Torsten
  Schaub, and Marius~Thomas Schneider.
\newblock Potassco: The {P}otsdam answer set solving collection.
\newblock {\em {AI} Commun.}, 24(2):107--124, 2011.

\bibitem{muta}
Jeron Kazius, Ross McGuire, and Roberta Bursi.
\newblock Derivation and validation of toxicophores for mutagenicity
  prediction.
\newblock {\em J. Med. Chem.}, 48(1):312--320, 2005.

\bibitem{lee12vldb}
Jinsoo Lee, Wook{-}Shin Han, Romans Kasperovics, and Jeong{-}Hoon Lee.
\newblock An in-depth comparison of subgraph isomorphism algorithms in graph
  databases.
\newblock {\em {PVLDB}}, 6(2):133--144, 2012.

\bibitem{lerouge17pr}
Julien Lerouge, Zeina Abu{-}Aisheh, Romain Raveaux, Pierre H{\'{e}}roux, and
  S{\'{e}}bastien Adam.
\newblock New binary linear programming formulation to compute the graph edit
  distance.
\newblock {\em Pattern Recognition}, 72:254--265, 2017.

\bibitem{mckay81cn}
B.~D. McKay.
\newblock Practical graph isomorphism.
\newblock {\em Congressus Numerantium}, 30:45--87, 1981.

\bibitem{mckay2014jsc}
Brendan~D. {McKay} and Adolfo Piperno.
\newblock Practical graph isomorphism, {II}.
\newblock {\em Journal of Symbolic Computation}, 60:94–112, 2014.

\bibitem{pasquier17socc}
Thomas Pasquier, Xueyuan Han, Mark Goldstein, Thomas Moyer, David~M. Eyers,
  Margo Seltzer, and Jean Bacon.
\newblock Practical whole-system provenance capture.
\newblock In {\em Proceedings of the 2017 Symposium on Cloud Computing (SoCC
  2017)}, pages 405--418, 2017.

\bibitem{riesen15}
Kaspar Riesen.
\newblock {\em Structural Pattern Recognition with Graph Edit Distance -
  Approximation Algorithms and Applications}.
\newblock Springer, 2015.

\bibitem{zampelli05cp}
St{\'{e}}phane Zampelli, Yves Deville, and Pierre Dupont.
\newblock Approximate constrained subgraph matching.
\newblock In {\em Proceedings of the 11th International Conference on
  Principles and Practice of Constraint Programming ({CP} 2005)}, pages
  832--836, 2005.

\end{thebibliography}

\newpage
\appendix
\section{Proofs}\label{app:proofs}

\begin{proof}[of Theorem~\ref{thm:correctness}]
  We outline the proof of part 1.  We slightly abuse notation by
considering $h$ both as a function and a binary relation (i.e. the
graph of the function).  If $h$ is a homomorphism, then for any node
$v$, $h(v) = w $ holds for exactly one $w$ whose label must match that
of $v$, so line 1 of the homomorphism specification holds.  Likewise,
for any edge $e$ where $\mathtt{e}(e,v_1,w_1,l)$ and $h(v_1) = v_2$
and $h(w_1) = w_2$, then  $h(e) = e'$ where
$src(e') = h(v_1) = v_2$ and $tgt(e') = h(w_1)= w_2$ and $label(e') =l =
label(e)$.  Therefore there is exactly one $e'$ such that $h(e,e')$
holds and $\mathtt{e}_2(e',v_2,w_2,l)$.  Finally, by similar
reasoning, because $h$ preserves property values, line 3 holds.  Conversely, if the specification holds
of $h$, then line 1 implies that $h$ is a total, label-preserving
function on node identifiers, and line 2 implies that $h$ is a total, source,
target, and label-preserving function on edge identifiers.  Finally,
again line 3 implies that $h$ preserves property keys and values.

The proofs of part 2 and 3 are
straightforward: for part 2 it suffices to observe that $h$ is a
homomorphism from $G_1$ to $G_2$ (read in the forward direction) and
from $G_2$ to $G_1$ (in the backward direction) if and only if it is an
isomorphism.  For part 3, by part 1 the solution must be a graph
homomorphism and the additional constraint ensures that it is
injective, witnessing subgraph isomorphism.

\end{proof}

We prove Lemma~\ref{lem:canonical} as the third part of the following lemma:
\begin{lemma}\label{lemma}
\begin{enumerate}
\item If $ops$ maps $G_1$ to $G_2$ and $ops \longrightarrow ops'$ then
  $ops'$ maps $G_1$ to $G_2$ and $|ops'| \leq |ops|$.
\item If $ops$ is a canonical edit script mapping $G_2$ to
  $G_3$ and $op^*$ is an edit operation mapping $G_1$ to $G_2$ then
  $op^*;ops $ rewrites to a canonical edit script $ops'$ mapping $G_1$ to $G_3$
  with
  $|op;ops| \leq |ops'|$.
\item If $ops$ is an edit script mapping $G_1$ to $G_2$, then there is
  a canonical edit script $ops'$ mapping $G_1$ to $G_2$ such that
  $|ops'| \leq |ops|$.
\end{enumerate}
\end{lemma}
\begin{proof}
  \begin{enumerate}
\item The proof is straightforward for each rule; in most cases, the
  two operations commute.  The interesting cases are:
  \begin{itemize}
  \item $\updP^*/\delP$: If the updated property is
    immediately deleted, it has the same effect as just deleting.
  \item $\insP^*/\updP$: If the inserted property is
    immediately updated, it has the same effect as just inserting the
    updated value.
\item $\insV^*/\delV$, $\insE^*/\delE$, $\insP^*/\delP$: If a node,
  edge, or property is inserted and immediately deleted,
  the two operations cancel out.
  \end{itemize}
\item Let $ops$ be a canonical edit script mapping $G_2$ to $G_3$, and
  $op^*$ a marked edit operation mapping $G_1$ to $G_2$.  Given an
  edit script with at most one marked operation, define the
  *-length of an edit script to be 0 if it contains no marked operation
  and $|op_1^*;ops_2|$ if it is of the form $ops_0;op_1^*;ops_2$. That is,
  the *-length is the length of the marked suffix of the edit script,
  or 0 if there is no mark.  All of the rules in
  Fig.~\ref{fig:rewrite} decrease the *-length of the edit script, as
  well as preserving or decreasing the length, so we can rewrite
  $op^*;ops$ to a normal form.  Moreover, clearly the rewrite rules
  preserve the order of the operations aside from the marked one, and
  in the cases where the mark is removed, the edit operation is in a
  position that is allowed in a canonical edit script (because all of
  the cases where a marked operation violates canonical form are
  covered by other rules).  Thus, after rewriting to a normal form,
  $op^*;ops$ is a canonical edit script.
  \item We proceed by induction on the length of $ops$. If it is
    empty, there is nothing to prove.  Otherwise, suppose it is of the
    form $op,ops_0$.  By induction, there must exist $ops_0' $
    equivalent to $ops_0$ with $|ops_0'| \leq |ops_0|$.  Consider the
    marked edit script $op^*;ops_0$.  By part 1,
    this normalizes to an unmarked edit script $ops'$ that is equivalent to $op;ops_0'$ with
    $|ops'| \leq |op^*;ops_0'| = |op;ops_0'| \leq |op;ops_0|= |ops|$.
  \end{enumerate}
\end{proof}

\section{Graph edit distance contest problem}
\label{sec:gedc}

Listing~\ref{fig:gedc-spec} shows the specification of the weighted
edit distance problem corresponding to the second set of parameter
settings in the Graph Edit Distance Contest~\cite{abuaisheh17prl}.  The first six
lines define the costs as constants.  Next the usual specification
that $h$ is a partial isomorphism (ignoring labels) is given.  Finally
we specify the costs associated with node and edge relabeling,
insertion, and deletion.  (It turns out to improve performance
slightly to omit the intermediate insertion, deletion and relabeling
predicates shown in Listings~\ref{fig:ged} and~\ref{fig:ged-relabel}.)

\begin{figure}[tb!]
  \lstset{frame=single,numbers=left,caption={Graph edit distance
      contest problem},label={fig:gedc-spec}}
\begin{lstlisting}
#const c_node_sub=2.
#const c_node_ins=4.
#const c_node_del=4.
#const c_edge_sub=1.
#const c_edge_ins=2.
#const c_edge_del=2.

{ h(X,Y) : n2(Y,_) } <= 1 :- n1(X,_).
{ h(X,Y) : n1(X,_) } <= 1 :- n2(Y,_).
{ h(X,Y) : e2(Y,S2,T2,_), h(S1,S2), h(T1,T2) } <= 1 :- e1(X,S1,T1,_).
{ h(X,Y) : e1(X,S1,T1,_), h(S1,S2), h(T1,T2) } <= 1 :- e2(Y,S2,T2,_).

node_cost(X,c_node_sub) :- n1(X,L1), h(X,Y), n2(Y,L2), L1 <> L2.
node_cost(Y,c_node_ins) :- n2(Y,L), not h(_,Y).
node_cost(X,c_node_del) :- n1(X,_), not h(X,_).

edge_cost(X,c_edge_sub) :- e1(X,_,_,L1), h(X,Y), e2(Y,_,_,L2), L1 <> L2.
edge_cost(Y,c_edge_ins) :- e2(Y,_,_,_), not h(_,Y).
edge_cost(X,c_edge_del) :- e1(X,_,_,_), not h(X,_).

#minimize { NC,X : node_cost(X,NC);
             EC,X : edge_cost(X,EC)}.
\end{lstlisting}
\end{figure}
\section{Approximate subgraph isomorphism}
\label{app:code}

\begin{figure}[tb!]
  \lstset{frame=single,numbers=left,caption={Approximate subgraph
      isomorphism (old)},label={fig:asgi-old}}
\begin{lstlisting}
{ h(X,Y) : n2(Y,_)} = 1 :- n1(X,_).
{ h(X,Y) : e2(Y,_,_,_)} = 1 :- e1(X,_,_,_).
:- X <> Y, h(X,Z), h(Y,Z).
:- X <> Y, h(Z,Y), h(Z,X).
:- n1(X,L), h(X,Y), not n2(Y,L).
:- e1(E1,_,_,L), h(E1,E2), not e2(E2,_,_,L).
:- e1(E1,X1,_,_), h(E1,E2), e2(E2,Y1,_,_), not h(X1,Y1).
:- e1(E1,_,X2,_), h(E1,E2), e2(E2,_,Y2,_), not h(X2,Y2).

#minimize { LC,X,K : prop_cost(X,K,LC) }.
prop_cost(X,K,0) :- p1(X,K,V),  h(X,Y), p2(Y,K,V).
prop_cost(X,K,1) :- p1(X,K,V1), h(X,Y), p2(Y,K,V2), V1 <> V2.
prop_cost(X,K,1) :- p1(X,K,V),  h(X,Y), not p2(Y,K,_).
\end{lstlisting}
  \lstset{frame=single,numbers=left,caption={Approximate subgraph
      isomorphism (new)},label={fig:asgi-new}}
\begin{lstlisting}
{h(X,Y) : n2(Y,L)} =  1 :- n1(X,L).
{h(X,Y) : n1(X,L)} <= 1 :- n2(Y,L).
{h(X,Y) : e2(Y,S2,T2,L), h(S1,S2), h(T1,T2)}  = 1 :- e1(X,S1,T1,L).
{h(X,Y) : e1(X,S1,T1,L), h(S1,S2), h(T1,T2)} <= 1 :- e2(Y,S2,T2,L).

prop_cost(X,K,0) :- p1(X,K,V),  h(X,Y), p2(Y,K,V).
prop_cost(X,K,1) :- p1(X,K,V1), h(X,Y), p2(Y,K,V2), V1 <> V2.
prop_cost(X,K,1) :- p1(X,K,V),  h(X,Y), not p2(Y,K,_).
#minimize { LC,X,K : prop_cost(X,K,LC) }.
\end{lstlisting}
\end{figure}

Two specifications defining approximate subgraph isomorphism are shown
in Listings~\ref{fig:asgi-old} and~\ref{fig:asgi-new}.  The code in
Listing~\ref{fig:asgi-old} is the same as that given by Chan et
al.~\cite{chan19mw}.  The code in Listing~\ref{fig:asgi-new} is a
variant of Listing~\ref{fig:sgi} that removes the constraint that
properties match exactly and instead associates a cost to each
property of $G_1$ not matched in $G_2$, which must be minimized.  This
is the same approach as followed in Listing~\ref{fig:asgi-old}; the
differences are in how the label preservation and edge preservation
constraints are defined.  Chan et al.~\cite{chan19mw} encoded these
constraints using several (possibly redundant) clauses, whereas in our
version, the functionality, label preservation, and edge preservation
are captured by just four constraints.

\end{document}